\documentclass[superscriptaddress, reprint, amsmath,amssymb,aps,nofootinbib]{revtex4-1}

\usepackage{bbm}
\usepackage{amssymb}
\usepackage{graphicx}
\usepackage{graphics}
\usepackage{afterpage}
\usepackage[abs]{overpic}
\usepackage{amsmath}
\usepackage{amsthm}
\usepackage{color}
\usepackage{dsfont}
\usepackage{appendix}
\usepackage{xcolor}
\usepackage{ucs}
\usepackage{graphicx}
\usepackage{url}

\usepackage{xcolor}
\usepackage{lipsum}

\usepackage{mathtools}
\usepackage{comment}

\usepackage{hyperref}
\definecolor{darkred}  {rgb}{0.5,0,0}
\definecolor{darkblue} {rgb}{0,0,0.5}
\definecolor{darkgreen}{rgb}{0,0.5,0}
\hypersetup{
  colorlinks = true,
  urlcolor  = blue,         
  linkcolor = darkblue,     
  citecolor = darkgreen,    
  filecolor = darkred       
}

\def\>{\rangle}
\def\<{\langle}

\def\mN{\mathcal{N}}

\def\mS{\mathcal{S}}

\def\mT{\mathcal{T}}
\def\mV{\mathcal{V}}

\renewcommand{\qedsymbol}{\nobreak \ifvmode \relax \else
	\ifdim \lastskip<1.5em \hskip-\lastskip \hskip1.5em plus0em
	minus0.5em \fi \nobreak \vrule height0.75em width0.5em
	depth0.25em\fi}

\renewcommand{\geq}{\geqslant}
\renewcommand{\leq}{\leqslant}

\newtheorem{theorem}{Theorem}
\newtheorem*{theorem*}{Theorem}

\newtheorem{lemma}{Lemma}
\newtheorem*{lemma*}{Lemma}

\newtheorem{definition}{Definition}
\newtheorem*{definition*}{Definition}

\theoremstyle{remark}
\newtheorem{remark}{Remark}

\theoremstyle{definition}

\newcommand{\bea}{\begin{eqnarray}}
\newcommand{\eea}{\end{eqnarray}}
\newcommand{\be}{\begin{equation}}
\newcommand{\ee}{\end{equation}}
\newcommand{\ba}{\begin{equation}\begin{aligned}}
\newcommand{\ea}{\end{aligned}\end{equation}}

\def\be{\begin{equation}}
\def\ee{\end{equation}}

\newcommand{\pr}{{\rm Prob}}

\newcommand{\mM}{\mathcal{M}}

\newcommand{\mbb}[1]{\mathbb{#1}}

\newcommand{\bra}[1]{\langle #1|}
\newcommand{\ket}[1]{|#1\rangle}

\newcommand{\eqdef}{\coloneqq}

\def\p{\mathbf{p}}
\def\q{\mathbf{q}}
\def\e{\mathbf{e}}
\def\r{\mathbf{r}}

\def\v{\mathbf{v}}
\def\t{\mathbf{t}}
\def\s{\mathbf{s}}
\def\e{\mathbf{e}}
\def\b{\mathbf{b}}

\newcommand{\da}{\downarrow}

\begin{document} 
\title{What is Entropy? \\
A new perspective from games of chance} 
	
\author{Sarah \surname{Brandsen}}\email{sarah.brandsen@duke.edu}
\address{
Department of Physics,
Duke University, 
Durham, NC, USA 27708}
\address{
Department of Mathematics and Statistics, Institute for Quantum Science and Technology,
University of Calgary, AB, Canada T2N 1N4}
\author{Isabelle Jianing \surname{Geng}}
\address{
Department of Mathematics and Statistics, Institute for Quantum Science and Technology,
University of Calgary, AB, Canada T2N 1N4}
\author{Gilad \surname{Gour}}
\address{
Department of Mathematics and Statistics, Institute for Quantum Science and Technology,
University of Calgary, AB, Canada T2N 1N4}

\begin{abstract}
Given entropy's central role in multiple areas of physics and science, one important task is to develop a systematic and unifying approach to defining entropy. Games of chance become a natural candidate for characterising the uncertainty of a physical system, as a system's performance in gambling games depends solely on the uncertainty of its output. In this work, we construct families of games which induce pre-orders corresponding to majorization, conditional majorization, and channel majorization. Finally, we provide operational interpretations for all pre-orders, show the relevance of these results to dynamical resource theories, and find the only asymptotically continuous classical dynamic entropy. 
\end{abstract}
\maketitle

{\it Introduction.} Entropy plays a central role in many areas of physics and science including statistical mechanics, thermodynamics, information theory, black hole physics, cosmology, chemistry, and even economics~\cite{entropy_econ, Bekenstein1, Bekenstein2}. Consequently, there are multiple approaches to understanding entropy: in thermodynamics it can be understood as a measure of energy dispersal at a given temperature, whereas in information theory entropy is a compression rate. Other attributes to entropy such as disorder, chaos,  randomness of a system, and the arrow of time~\cite{arrow_of_time}, have also been studied extensively in literature.  These different attributes and contexts also leads to different measures of entropy such as the Gibbs and Boltzmann entropy~\cite{Jaynes}, Tsallis entropies~\cite{Tsallis}, R\'enyi entropies~\cite{Renyi}, and the von-Neumann~\cite{VonNeumann} and Shannon entropies~\cite{Shannon}, and other entropy functions such as molar entropy~\cite{MolarEntropy}, entropy of mixing~\cite{MixingEntropy}, loop entropy~\cite{LoopEntropy}, and so on.

The diverse roles of entropy cry out for a more systematic and unifying approach, in which entropy is defined rigorously and in a way that is independent of the physical context. We identify \emph{uncertainty} as the key trait of entropy that is common in all of its expressions. It could be the uncertainty about the state of a physical system, the uncertainty of a source in a compression scheme, or uncertainty corresponding to disorder and randomness.

Games of chance~\cite{guessinggames} are ideal candidates for studying uncertainty, as the performance in such games depends solely on the certainty about the outcome of the game such that ``more uncertain'' systems will be less likely to win. Thus, for any physical object such as a system, measurement, or channel, we define its degree of uncertainty in terms of the probability of winning a game of chance. Since we construct a family of games, uncertainty cannot be quantified with a single function, but rather is characterized with a partial order where system $A$ is said to be ``less uncertain'' than system $B$ if system $A$ performs at least as well as system $B$ for \emph{all} games of chance. 

In this paper we construct gambling games that give rise to three types of partial orders: majorization, conditional majorization, and channel majorization. The first two partial orders characterize the degree of uncertainty and conditional uncertainty in (possibly composite) physical systems, while the last characterizes the uncertainty associated with a channel. Crucially, we additionally provide an operational characterization for all partial orders: for example, we demonstrate that conditional majorization is equivalent to a relation induced by a conditional random relabeling map.

Our work extends previous results about the entropy of channels~\cite{Gour3, Yuan, Fang, Devetak, DDI, Brandsen}, and the induced partial orders we find are consistent with intuition as well as operationally motivated. In the case of channel majorization, we find the unique asymptotically continuous classical channel entropy function. Additionally, we provide a first operational interpretation to the complete family of dynamical monotones introduced in~\cite{Gour2}. Finally, the wide applicability of games of chance opens up opportunities to extend our work to the quantum world and operationally characterize the entropy of a quantum channel.



{\it Dice Games and Majorization.} Gambling games are games in which a player is provided with partial information and use statistical inference to take their best guess in the face of incomplete information. We first consider a gambling game in which the host rolls a biased dice, and the player has to guess its outcome. Denote by $\p=(p_1,...,p_d)^T$ the probability vector corresponding to the $n$ possible outcomes, and denote by $\p^{\downarrow}=(p_1^{\downarrow},...,p_d^{\downarrow})^T$ the vector obtained from $\p$ by rearranging its components in non-increasing order. For simplicity, we will always assume that the components of any probability vectors $\p$ are arranged in non-increasing order so that $\p=\p^{\downarrow}$. 

A \emph{$w$-gambling game} occurs when the player is allowed to provide a set with $w$-numbers as guesses prior to rolling the dice. The player then wins if the outcome from the dice roll belongs to the set of guesses. For example, if $w=2$, then the player will choose to provide numbers $\{1, 2\}$ (as these have the highest probability of occurring), and will win the game with probability $p_1^{\downarrow}+p_2^{\downarrow}$. In general, the probability of winning a $w$-game with dice $\mathbf{p}$ can be denoted as:

\be
\pr_w(\p)=\|\p\|_{(w)}\eqdef\sum_{x=1}^{w}p_x^{\downarrow}
\ee
where $\|\cdot\|_{(w)}$ denotes the Ky-Fan norm. 

Suppose that at the beginning of each game, the player is allowed to choose between two dice with corresponding probabilities $\p$ and $\q$. Clearly, the player will choose the dice which gives better odds of winning the game, and so will choose the $\p$-dice if $\|\p\|
_{(w)}\geq\|\q\|_{(w)}$. In general, the player's choice will depend on the value of $w$- for example, if $\p = (\frac{1}{2}, \frac{1}{2}, 0)$ and $\q = (\frac{2}{3}, \frac{1}{6}, \frac{1}{6})$ then the player will choose $\q$ when $w=1$ and $\p$ when $w=2$.

If the player knows the distribution $\{t_w\}_{w=1}^{m}$ from which the $w$-gambling game is determined, then the probability that the player wins such a $w$-game is given by 
\be
\pr_\t(\p)=\sum_{k=1}^{m}t_k\|\p\|_{(k)} :=\sum_{x=1}^{m}\sum_{k=x}^{m}t_kp_x\equiv\sum_{x=1}^{m}r_xp_x=\r\cdot\p\;,
\ee
where the vector $\r$ with components $\{r_x\}$ is given by $\r\equiv U\t$ and $U$ is the $m\times m$ invertible upper triangular matrix
\be\label{ut}
U\equiv\begin{pmatrix}
1 & 1  & \cdots & 1\\
0 & 1  & \cdots & 1\\
\vdots &  \vdots & \ddots & \vdots\\
0 & 0 & \cdots & 1
\end{pmatrix}\quad
\ee
Note that the vector $\r$ satisfies $\r=\r^\da$ and that $\r$ has this property if and only if the vector $\t=U^{-1}\r$ has non-negative components. In general, we allow $t\equiv\sum_w t_w$ to be strictly smaller than one, in which case there is a non-zero probability that the player loses the game irrespective of the dice outcome. We likewise set $p^\da_x \equiv 0$ if $x>d$. 

Finally, we say that $\p$ majorizes $\q$ and write $\q \precsim \p$ if and only if 
\be\label{a22}
\pr_\t(\q)\leq\pr_\t(\p)\;.
\ee 
for all (possibly incomplete) distributions $\{t_w\}$. This can be interpreted as stating that $\q \precsim \p$ if the player will \emph{always} choose the $\p$-dice over the $\q$-dice for any gambling game.

One can additionally consider games with resources, where a player may be provided with an additional $\s$-dice. 
The probability to win such a $w$-gambling game with resources is given by $\|\p\otimes\s\|_{(w)}$, such that a $\p$-dice has less uncertainty than a $\q$-dice if $\|\p\otimes\s\|_{(w)}\geq \|\q\otimes\s\|_{(w)}$ for all $w$ and for all $\s$. However, this condition is equivalent to $\q\precsim \p$, so the class of $k$-games are sufficient to induce the pre-order that is generated by the larger class of games with resources. 

{\it Conditional Majorization: Games with a correlated source.} Here we consider a game in which the host rolls a dice with two outcomes $x$ and $y$. The host sends the value of $x$ to the player, and the value $y$ is kept hidden from the player. The players knows the distribution $\{p_{xy}\}$ from which $x$ and $y$ are sampled, and the player's goal is to guess the value of $y$.
Our goal is to construct all possible gambling games that incorporate a correlated source, so we allow the player to choose a value $z$ and then have the host select $w$ from  a conditional distribution $\mT\eqdef\{t_{w|z}\}$ after receiving the value $z$ from the player. We denote the player's choice of $z$ with a function $z=f(x)$. In general, the player will choose $z$ based on their knowledge of $x$, as well as the fixed distributions $\{p_{xy}\}$ and $\{t_{w|z}\}$.  
In Fig.~\ref{classical1} we depict such a $\mathcal{T}$-gambling game.

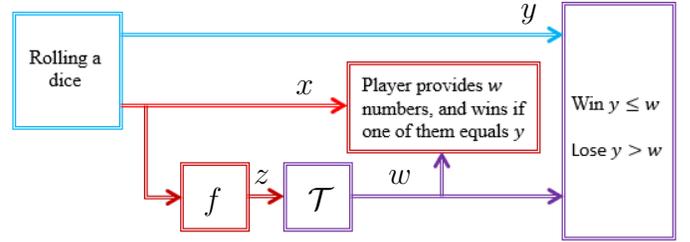
\begin{figure}[h]
\large
\begin{overpic}[scale=.78]{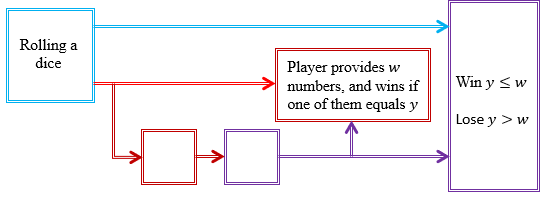}
\put(114.5,15.5){$\mathcal{T}$}
\put(75,15.5){$f$}
\put(110,59){$x$}
\put(195,88){$y$}
\put(94,25){$z$}
\put(145,25){$w$}
\end{overpic}
 \caption{\linespread{1}\selectfont{\small Classical gambling game with a correlated source. The player is provided with the value $x$. Based on this value, the player choose $z$ (or the function $f$) and send it to the host. The host then chooses the $w$ game based on a (possibly incomplete) distribution matrix $\mathcal{T}=(t_{w|z})$. The player will then guess $\{1, ..., w\}$ and will win the game if $y\leq w$.}} 
  \label{classical1}
\end{figure}

Let $P=(p_{xy})$ be the $m\times n$ probability matrix, and w.l.o.g. suppose $P = P^{\downarrow}$ such that
\be
p_{x1}\geq p_{x2}\geq\cdots\geq p_{xn}\quad\forall\;x=1,...,m.
\ee
For a given $x$ and $z$, the probability to win the game can be expressed as
$\r_z\cdot\p_x$, where $\{\p_x\}$ are the rows of $P$ and $\r_z\equiv U\t_z$, where $\{\t_z\}$ are the columns of the $\ell\times q$ matrix $\mathcal{T} =(t_{w|z})$.
As before, we only require that $\sum_{w=1}^{\ell}t_{w|z}\leq 1$ for all $z=1,...,q$. 
Therefore, the optimal probability to win a $\mathcal{T}$-game is given by
\be\label{pt}
\pr_\mathcal{T}(P)=\sum_{x=1}^{m}\max_{z}\r_{z}\cdot\p_x\;.
\ee
We are now ready to compare between two dice, a $P$-dice and a $Q$-dice, and call this comparison conditional majorization.

\begin{definition}\label{defconmaj}
Let $P=(p_{xy})$ be an $m\times n$ probability matrix, and $Q=(q_{x'y'})$ be an $m'\times n'$ probability matrix. 
We say that $P$ conditionally majorizes $Q$ and write
\be
Q\precsim_cP\quad\text{if and only if}\quad \pr_\mathcal{T}(Q)\leq \pr_\mathcal{T}(P)
\ee
for all (column) sub-stochastic matrices $\mathcal{T}$. \footnote[1]{
The term `conditional majorization' was first introduced in~\cite{Gour}, however it was not defined in an operational way as with the games of chance. }
\end{definition}

\begin{theorem}\label{classical8}
Let $P$ and $Q$ be two $m\times n$ column stochastic matrices. Then,
\be\label{class0}
Q\precsim_cP\quad\iff\quad
Q=\sum_zS_{z}PV_z
\ee 
where each $S_z$ is a sub-stochastic matrix such that $\sum_zS_z$ is a column stochastic matrix (i.e. a classical channel), and each $V_z$ is a permutation matrix.
\end{theorem}

\begin{remark}

This theorem provides an operational characterization for conditional majorization.\footnote[2]{For convenience we assume that the two matrices have the same dimensions, as if this is not the case one can add zero rows and columns to make them the same dimension. } It states that conditional majorization is equivalent to a relation induced by a conditional random relabeling map; see Fig~\ref{relabeling}.
\end{remark}

\begin{figure}[h]
   \includegraphics[width=0.5\textwidth]{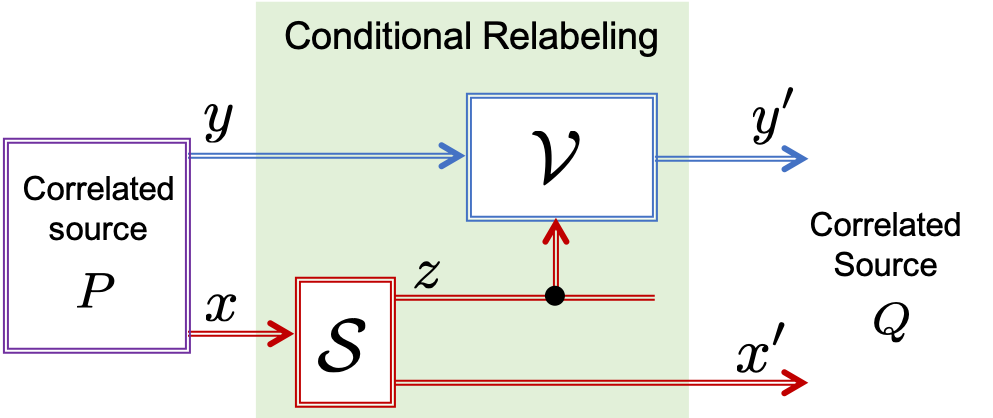}
  \caption{\linespread{1}\selectfont{\small The action of conditional random relabeling map on a correlated source $P=(p_{xy})$ yields the correlated source $Q=(q_{x'y'})=\sum_zS_zPV_z$.}} 
  \label{relabeling}
\end{figure}

\begin{proof}
The proof follows from the following two lemmas that were proven in~\cite{Gour}. For completeness, we provide the proof of these lemmas in appendix A of the supplementary material.
\begin{lemma}~\cite{Gour}
The matrices $P$ and $Q$ are related as in the RHS of~\eqref{class0} if and only if there exists a column stochastic matrix $S$ such that
\be\label{npu2}
QU\leq SPU\;,
\ee
where the inequality is entry-wise and $U$ is the upper triangular matrix 
\end{lemma}

\begin{lemma}\cite{Gour}
There exists a column stochastic matrix $S$ that satisfies~\eqref{npu2} if and only if for any set of  $m$ vectors $\r_1,...,\r_m\in\mbb{R}_{+}^{n}$ whose components are arranged in non-decreasing order,
 \be
\sum_{x=1}^{m}\max_z\p_{x}\cdot\r_z\geq \sum_{x=1}^{m}\q_x\cdot\r_x\;.
\ee
\end{lemma}
We now prove the theorem.
If $Q\precsim_c P$ then from~\eqref{pt} and Definition~\ref{defconmaj} we get
\be
\sum_{x=1}^{m}\max_z\p_{x}\cdot\r_z \geq \sum_{x=1}^{m}\max_z\q_{x}\cdot\r_z\geq \sum_{x=1}^{m}\q_{x}\cdot\r_x\;.
\ee
Hence, from  the above two lemmas  it follows that $P$ and $Q$ are related as in the RHS of~\eqref{class0}. 
 \end{proof}

{\it Channel Majorization: Games with a classical channel.} We now identify the class of gambling games corresponding to the notion of entropy of a channel. Roughly speaking,  our aim is to identify games in which the player has a lower probability of winning the game with a noisier channel. We will consider a classical channel $\mM$ with transition matrix $P=(p_{y|x})$, and we denote the space of channels from system $A$ to system $B$ as $\mathfrak{L}(\text{A} \rightarrow \text{B})$. The goal of the game is for the player to correctly guess the value $y$ at the output of the channel. 

In the most general settings, the Host does not provide the player the full information about $w$ at the early stage of the game. Instead, the player receives a number $z$ that is sampled from a distribution $\{t_{wz}\}$ with $w=1,...,m$ and $z=1,...,\ell$.
The player knows the $m\times\ell$ probability matrix $T=(t_{wz})$. Based on this partial information about which $w$-game will be played later on, the player will choose the optimal value of $x$ to send through the channel. Finally, the host draws the value of $w$ and the player wins if $y \leq w$.
Such a $T$-gambling game with channel $\mM$ is then depicted in Fig.~\ref{classical1}.

\begin{figure}[h]
\large
\begin{overpic}[scale=.78]{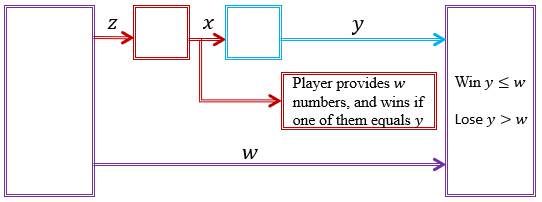}
\put(20,48){$T$}
\put(112,75){$\mathcal{M}$}
\put(71,75){$\mathcal{E}$}
\end{overpic}
\caption{\linespread{1}\selectfont{\small A classical gambling game with a channel. The host provides the player with a value $z$ that is drawn from a probability distribution $T=(t_{wz})$. The player chooses an input to the channel $x=f(z)$ that may depends of the value of $z$.
The host then selects $w$, and the player wins if $y\leq w$ since they will always provide the first $w$ numbers with the highest corresponding probabilities.}} 
  \label{channel1}
\end{figure}

Let $P=(p_{y|x})=P^{\da}$ be the ordered transition matrix of the channel $\mM$, in which the columns of $P$ are arranged in non-increasing order.
For a given choice of $x$ and $z$ the probability that the player wins the game is given by
\be
\sum_{w=1}^{m}t_{w|z}\sum_{y=1}^{w}p_{y|x}=\sum_{y=1}^{m}\sum_{w=y}^{m}t_{w|z}p_{y|x}\equiv  \frac{\r_z\cdot\p_{x}}{|\mathbf{t}_{z}|}
\ee
where $\r_z\equiv U\t_z$ (with $\{\t_z\}$ being the columns of $T$), and $\{\p_x\}$ are the columns of the transition matrix $P$; and in particular, each $\p_x$ is a probability vector. 
For each value of $z$ the player will choose $x$ (i.e. $f(z)$) such that $\r_z\cdot\p_x=\max_{x'}\r_z\cdot\p_{x'}$. We therefore conclude that the optimal probability to win a $T$-gambling game with a classical channel $\mM=P$ is given by
\be
\pr_{T}(\mM)=\sum_{z=1}^{\ell}\max_{x}\;\r_z\cdot\p_{x}
\ee
Note that the above quantity is the dual of~\eqref{pt} in the sense that the role of $x$ and $z$ are flipped. Unlike~\eqref{pt}, here $\p_x$ is a probability vector for each $x$.

\begin{definition}\label{defchanmaj}
Let $\mM$ and $\mN$ be two classical channels with corresponding transition matrices $P=(p_{y|x})$ and
$Q=(q_{y'|x'})$. 
We say that the channel $\mN$  majorizes $\mM$ and write
\be
\mM\precsim \mN\quad\text{if and only if}\quad \pr_T(\mM)\leq \pr_T(\mN)
\ee
for all (column) sub-stochastic matrices $T$.
\end{definition}

We now provide the following characterization of channel majorization.
\begin{theorem}\label{cun}
Let $\mM^{X \to Y}$ and $\mN^{X'\to Y'}$ be two classical channels. Then, $\mM\precsim \mN$ if and only if there exists a channel 
$\mS^{X \to X'W'}$ and a controlled isometry $\mV^{W'Y' \to Y}$ (i.e. for each $w'$, $\mV^{W'Y' \to Y}(w',y')$ is an injective function of $y'$) such that
\be\label{rel0}
\mM^{X \to Y}=\mV^{W'Y' \to Y}\circ\mN^{X'\to Y'}\circ\mS^{X \to X'W'}
\ee
(see Fig.~\ref{channel2}).
\end{theorem}

\begin{proof}
See appendix B of the supplementary material for the complete proof. We note here that the proof additionally provides the first operational interpretation of the dynamical monotones introduced in~\cite{Gour2}. 
\end{proof}

\vspace{0 cm}
\begin{figure}[h]
\large
\begin{overpic}[scale=.74]{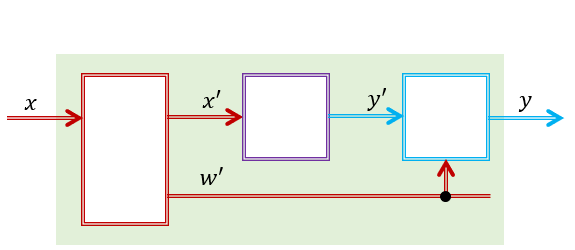}
\put(50,42){$\mathcal{S}$}
\put(121,56){$\mathcal{N}$}
\put(191,56){$\mathcal{V}_{w'}$}
\put(75, 95){The Channel $\mathcal{M}$}
\end{overpic}
\caption{\linespread{1}\selectfont{\small The simulation of $\mM$ with $\mN$ in the case that $\mM\precsim \mN$.}} 
  \label{channel2}
\end{figure}

We now provide a lemma showing that it is sufficient to consider only games with a fixed value of $z$. 
\begin{lemma}
The pre-order induced by the set of classical gambling games is unchanged if we restrict to the set of gambling games such that $\mathcal{T} = \mathbf{p}$ where $\mathbf{p}$ is a vector probability distribution. Equivalently, $\mathcal{M} \precsim \mathcal{N}$ if and only if $\text{Prob}_{\mathbf{p}}(\mathcal{M}) \leq \text{Prob}_{\mathbf{p}}(\mathcal{N})$ for all $\mathbf{p}$.
\end{lemma}
\begin{proof}
See appendix C of the supplementary material for the complete proof. 
\end{proof}
Equipped with the above pre-order, we now provide an operational definition of the family of entropy functions. 
\begin{definition} (cf.~\cite{Gour3, Gour5})
A non-zero function $$H : \underset{\text{A}, \text{B}}{\bigcup} \mathfrak{L}(\text{A} \rightarrow \text{B}) \rightarrow \mathbb{R}$$ is a channel entropy if it satisfies the following two conditions: 
\begin{enumerate}
    \item It is monotonic under the channel majorization; i.e.
\begin{align*}
    H(\mathcal{M}) \geq H(\mathcal{N}) \ \ \text{if} \ \ \mathcal{M} \precsim \mathcal{N}
    \end{align*}
\item It is additive under tensor products; i.e. 
\be
H(\mN\otimes\mM)=H(\mN)+H(\mN)
\ee
for all $\mN\in\mathfrak{L}(A\to B)$ and $\mM\in\mathfrak{L}(A\to B)$.
\end{enumerate}
\end{definition}
\noindent \emph{Remark 2.} Due to Theorem~\ref{cun}, the above definition can be shown to be equivalent to the definition of entropy for classical channels previously outlined in~\cite{Gour3, Gour5} (see appendix D of supplementary material for details).
Note however that unlike the definitions in~\cite{Gour3, Gour5}, the monotonicity property above is motivated operationally by games of chance.

Finally, we prove that there is only one unique dynamical entropy function that reduces to the Shannon entropy on classical states. 

\begin{theorem}
Let $H$ be a classical dynamical entropy that reduces to the Shannon entropy, $H_{S}$, on classical states. Then for all $\mathcal{N} \in \mathfrak{L}({\text{A} \rightarrow \text{B}})$, 
\begin{align*}
    H(\mathcal{N}) &= \min_{x} H_{S}\Big(\mathcal{N}(\ket{x}\bra{x})\Big)
\end{align*}
where in the above we represent the state corresponding to $x$ in Dirac notation as $\ket{x}\bra{x}$. 
\end{theorem}
\begin{proof}
See appendix E of the supplementary material for the complete proof. 
\end{proof}
\noindent \emph{Remark 3.} It follows from~\cite{Andreas_Winter} that $H(\mathcal{N})$ is asymptotically continuous (see appendix E of supplementary material for details). \\
{\it Conclusions.} In this work, we motivate the use of payoff functions from games of chance as a measure of uncertainty. From this, we introduce families of games of chance which give rise to three different partial orders: majorization, conditional majorization, and channel majorization. We showed that conditional majorization can be operationally interpreted as a relation induced by a conditional random relabeling map, while channel majorization can be operationally interpreted as a relation induced by applying an arbitrary pre-processing channel and (possibly correlated) random post-processing isometries. Finally, we find the only asymptotically continuous channel entropy. 

One natural extension of this work is to characterise the uncertainty of \emph{quantum} channels. Unlike the classical case, we expect the ordering induced by quantum gambling games to also take into account resources such as entanglement. An open question is whether there also exists a unique asymptotically continuous channel entropy in the quantum case. Finally, we demonstrated preliminary connections between our results and dynamical resource theories which can be further explored. 

{\it Acknowledgments.} The authors would like to thank Henry Pfister and Mark Wilde for helpful discussions.  GG and IG acknowledge support from the Natural Sciences and Engineering Research Council of Canada (NSERC). SB acknowledges support from the National Science Foundation (NSF) under Grant No. 1908730 and 1910571.  Any opinions, findings, conclusions, and recommendations expressed in this material are those of the authors and do not necessarily reflect the views of these sponsors.

{\it Author Contributions.} G.G. conceived the presented idea. S.B. extended the results. S.B. and I.G. verified all results and finalized the manuscript.

\bibliographystyle{ieeetr}


\appendix
\section{Proof of Theorem 1}
{\bf Theorem 1.} \emph{Let $P$ and $Q$ be two $m\times n$ column stochastic matrices. Then,
\be\label{class0}
Q\precsim_cP\quad\iff\quad
Q=\sum_zS_{z}PV_z
\ee 
where each $S_z$ is a sub-stochastic matrix such that $\sum_zS_z$ is a column stochastic matrix (i.e. a classical channel), and each $V_z$ is a permutation matrix.}

\begin{lemma}~\cite{Gour}
The matrices $P$ and $Q$ are related as in the RHS of~\eqref{class0} if and only if there exists a column stochastic matrix $S$ such that
\begin{align*}
QU \leq SPU 
\end{align*}
where the inequality is entry-wise and $U$ is the upper triangular matrix 
\end{lemma}
To prove the lemma, denote by $\{\p_x\}_{x=1}^{n}$ and $\{\q_{w}\}_{w=1}^{n}$ the $n$ rows of the matrices $P$ and $Q$, and by $s_{wz|x}$ and $s_{w|x}$ the components of the matrices $S_z$ and $\sum_zS_z$. Then, the RHS of~\eqref{class0} can be expressed as:
\be\label{eqn:qn}
\q_{w}=\sum_{x=1}^{n}s_{w|x}\p_xD_{wx}\quad\text{where}\quad D_{wx}\equiv\sum_{z}\frac{s_{wz|x}}{s_{w|x}}V_z\;.
\ee
Since $D_{wx}$ is a doubly-stochastic matrix we have that for each $x$ and $w$, $\p_xD_{wx}\prec\p_x$. We therefore conclude that
if the matrices $Q$ and $P$ are related as in~\eqref{qn} then there must exists a column stochastic matrix $S=(s_{w|x})$ such that
\be\label{qxp}
\q_{w}\precsim \sum_{x=1}^{n}s_{w|x}\p_x\quad\forall w=1,...,n.
\ee
Conversely, if the relation above holds, then there exists $n$ doubly-stochastic matrices $D_{w}$ such that 
\be
\q_{w}= \sum_{x=1}^{n}s_{w|x}\p_xD_{w}\quad\forall w=1,...,n.
\ee
Expressing each $D_{w}=\sum_{z}c_{z|w}V_z$ as a convex combination of all permutation matrices we obtain the form~\eqref{class}. Finally, recall that for each $w$ 
the row vectors $\p_x=\p_x^{\downarrow}$ and $\q_{w}=\q_{w}^{\downarrow}$ so that~\eqref{qxp} is equivalent to 
\be
\q_{w}U\leq \sum_{x=1}^{n}s_{w|x}\p_xU\quad\forall w=1,...,n,
\ee
where the inequality is entry-wise. This completes the proof of the lemma.

The question of whether there exists such a column stochastic matrix $S$ that satisfies~\eqref{npu0} can be solved with linear programming. In the following lemma, which was proved in~\cite{Gour}, we show that the dual problem can be expressed in terms of sub-linear functionals.

\begin{lemma}\cite{Gour}
There exists a column stochastic matrix $N$ that satisfies~\eqref{npu0} if and only if for any set of  $n$ vectors $\v_1,...,\v_n\in\mbb{R}_{+}^{m}$ whose components are arranged in non-decreasing order,
 \be
\sum_{w=1}^{n}\max_x\p_{w}\cdot\v_x\geq \sum_{w=1}^{n}\q_w\cdot\v_w\;.
\ee

\end{lemma}

To prove the lemma, denote by $\s_{x}$ the rows of $S$, so that~\eqref{npu} can be expressed as
\be\label{xpu}
\s_xPU\geq \q_xU
\ee
and the condition that $N$ is column stochastic is equivalent to $\sum_x\s_x=\e\equiv(1,...,1)$.
Denote further by $\s\equiv(\s_1,...,\s_n)$ an $n^2$-dimensional vector, and $\b\equiv (\q_1U,...,\q_nU,-\e)$ an $(nm+n)$-dimensional vector,  and by $A$ the $n^2\times n(m+1)$ matrix
\be
A\equiv
\begin{pmatrix}
PU & 0 & \cdots & 0 & -I_n\\
0 & PU & \cdots & 0 & -I_n\\
\vdots & \vdots & \ddots &\vdots &\vdots\\
 0 & 0 & \cdots & PU & -I_n\\
\end{pmatrix}
\ee
With these notations~\eqref{xpu} is equivalent to $\s A\geq \b$. Now, the question whether there exists $\s\geq 0$ that satisfies $\s A\geq\b$ is a feasibility problem in linear programming.
From the Farkas lemma, such a row vector $\s\in\mbb{R}_{+}^{n^2}$ exists if and only if for any column vector $\t\in\mbb{R}_{+}^{n(m+1)}$ that satisfies $A\t\leq 0$ we also have
$\b\cdot\t\leq 0$. Denote $\t=(\t_1,...,\t_n,\r)^T$, where for each $x=1,...,n$, $\t_x\in\mbb{R}_{+}^{m}$ and $\r\in\mbb{R}^n$. Then, with this notation the dual problem has the form
\be\label{51}
\r\geq PU\t_x\quad\forall x=1,...,n\quad\Rightarrow\quad\e\cdot\r\geq \sum_{x=1}^{n}\q_x\cdot(U\t_x)
\ee
Note that the condition  that $\r\geq PU\t_x$ for all $x$ can be expressed as $r_{w}\geq\max_x\p_{w}\cdot(U\t_x)$. Therefore, taking $r_{w}=\max_x\p_{w}U\t_x$ we conclude that~\eqref{51} holds if and only if
for any $n$ vectors $\t_1,...,\t_n\in\mbb{R}_{+}^{m}$
\be
\sum_{w=1}^{n}\max_x\p_{w}\cdot(U\t_x)\geq \sum_{x=1}^{n}\q_x\cdot(U\t_x)\;.
\ee
Denote by $\v_x\equiv U\t_x$ , and note that $\v_x=\v_x^{\downarrow}$. Moreover, the inverse of $U$ is given by
\be
U^{-1}=\begin{pmatrix}
1 & -1 & 0 & 0&\cdots & 0\\
0 & 1 & -1 & 0 &\cdots & 0\\
0 & 0 & 1 & -1&\cdots & 0\\
\vdots & \vdots & \vdots & \vdots & \ddots & \vdots\\
0 & 0 & 0 & 0 &\cdots & 1
\end{pmatrix}
\ee
Hence,  $\t_x\geq 0$ if and only if $\v_x=\v_x^{\downarrow}$. We therefore conclude that~\eqref{51} holds if and only if
for any $n$ vectors $\v_1,...,\v_n\in\mbb{R}_{+}^{m}$ whose components are arranged in non-decreasing order,
 \be
\sum_{w=1}^{n}\max_x\p_{w}\cdot\v_x\geq \sum_{x=1}^{n}\q_x\cdot\v_x\;.
\ee
This completes the proof of the lemma. 

\section{Theorem 2}
{\bf Theorem 2.} \emph{
Let $\mM^{X \to Y}$ and $\mN^{X'\to Y'}$ be two classical channels. Then, $\mM\precsim \mN$ if and only if there exists a channel 
$\mS^{X \to X'W'}$ and a controlled isometry $\mV^{W'Y' \to Y}$ (i.e. for each $z'$, $\mV^{W'Y' \to Y}(w',y')$ is an injective function of $y'$) such that
\be\label{rel0}
\mM^{X \to Y}=\mV^{W'Y' \to Y}\circ\mN^{X'\to Y'}\circ\mS^{X \to X'W'}
\ee
}
\begin{proof}

Let $P$ and $Q$ be the transition matrices corresponding to $\mN^{X\to Y}$ and $\mM^{X'\to Y'}$, respectively.
Note that w.l.o.g. we can assume that $|Y|=|Y'|\equiv m$, so that $P$ is an $m\times n$ matrix and $Q$ is an $m\times n'$ matrix.
Then, the relation~\eqref{rel0} is equivalent to
\be\label{rel1}
Q=\sum_{z=1}^{\ell}V_zPS_z
\ee
where $\sum_{z=1}^{\ell}S_z$ is a column stochastic matrix and each $V_z$ is a permutation matrix. 

\begin{lemma}
The matrices $P$ and $Q$ are related as in~\eqref{rel1} if and only if there exists a column stochastic matrix $S\equiv\sum_{z=1}^{\ell}S_z$ such that
\be\label{rel2}
U^TQ\leq U^TPS\;,
\ee
where the inequality is entry-wise and $U$ is the upper triangular matrix as defined in~\eqref{ut}.
\end{lemma}

To prove the lemma, denote by $\{\p_x\}_{x=1}^{n}$ and $\{\q_{x'}\}_{x'=1}^{n'}$ the $n$ and $n'$ columns of the matrices $P$ and $Q$, and by $s_{xz|x'}$ and $s_{x|x'}$ the components of the matrices $S_z$ and $S\equiv\sum_{z=1}^{\ell}S_z$. Then, the relation~\eqref{rel1} can be expressed as:
\be
\q_{x'}=\sum_{x=1}^{n}s_{x|x'}D_{xx'}\p_x\quad\text{where}\quad D_{xx'}\equiv\sum_{z=1}^{\ell}\frac{s_{xz|x'}}{s_{x|x'}}V_z\;.
\ee
Since for each $x$ and $x'$ the matrix $D_{xx'}$ is an $m\times m$ doubly-stochastic matrix, we have that for each $x$ and $x'$, $D_{xx'}\p_x\prec\p_x$. Thus,
if the matrices $Q$ and $P$ are related as in~\eqref{rel1} then there must exists a column stochastic matrix $S=(s_{x|x'})$ such that
\be\label{rel4}
\q_{x'}\prec \sum_{x=1}^{n}s_{x|x'}\p_x\quad\forall x'=1,...,n.
\ee
The above relation is equivalent to~\eqref{rel2} since for each $x$ and $x'$, $\q_{x'}=\q_{x'}^{\da}$ and $\p_x=\p_x^{\da}$.

Conversely, suppose~\eqref{rel2} holds. Therefore, also~\eqref{rel4} holds, so that there exists an $m\times m$ doubly-stochastic matrix $D$ such that 
\be
\q_{x'}= \sum_{x=1}^{n}s_{x|x'}D\p_x\quad\forall x'=1,...,n'.
\ee
Expressing $D=\sum_{z=1}^{\ell}c_{z}V_z$ as a convex combination of permutation matrices we obtain 
\be
\q_{x'}= \sum_{x=1}^{n}s_{x|x'}\sum_{z=1}^{\ell}c_{z}V_z\p_x=\sum_{x,z}V_z\p_x s_{x|x'}c_z\quad\forall x'=1,...,n.
\ee 
Finally, denote $s_{xz|x'}\equiv c_zs_{x|x'}$ and note that with this notation the above equation is equivalent to~\eqref{rel1}, where for each $z$ the the components of the matrix $S_z$  are $s_{xz|x'}$.
This completes the proof of the lemma.

The question of whether there exists such a column stochastic matrix $S$ that satisfies~\eqref{rel2} can be solved with linear programming.

\begin{lemma}
There exists a $n\times n'$ column stochastic matrix $S$ that satisfies~\eqref{rel2} if and only if for any set of  $n'$ vectors $\r_1,...,\r_{n'}\in\mbb{R}_{+}^{m}$ whose components are arranged in non-decreasing order,
 \be\label{lem0}
\sum_{x'=1}^{n'}\max_{x}\;\r_{x'}\cdot\p_{x}\geq \sum_{x'=1}^{n'}\r_{x'}\cdot\q_{x'}\;.
\ee
\end{lemma}

To prove the lemma, denote by $\{\s_{x'}\}_{x'=1}^{n'}$ the columns of $S$, so that~\eqref{rel2} can be expressed as
\be\label{e0}
U^TP\s_{x'}\geq U^T\q_{x'}\quad\quad\forall\;x'=1,...,n'\;.
\ee
and the condition that $S$ is column stochastic is equivalent to $\e\cdot\s_x= 1$ for each $x=1,...,n$, where $\e\equiv (1,...,1)^T\in\mbb{R}^n$.
Note however that it is sufficient to require that $\e\cdot\s_x\leq 1$ since by adding to $\s_x$ a non-negative vector, the relation~\eqref{e0} is preserved.
Denote further by $\e'\equiv (1,...,1)^T\in\mbb{R}^{n'}$,
\be
A \equiv
\begin{pmatrix}
U^TP & 0 & \cdots & 0 \\
0 & U^TP & \cdots & 0 \\
\vdots & \vdots & \ddots &\vdots\\
 0 & 0 & \cdots & U^TP\\
 -\e^T & 0 &\cdots & 0\\
 0 & -\e^T & \cdots & 0\\
 \vdots & \vdots & \ddots & \vdots\\
 0 & 0 & \cdots & -\e^T
\end{pmatrix}\in\mbb{R}^{(m+1)n'\times nn'} 
\ee
\be
\b \equiv
\begin{pmatrix}
U^T\q_1 \\
\vdots\\
U^T\q_{n'}\\
-\e'
\end{pmatrix}\in\mbb{R}^{(m+1)n'}, \\
\s \equiv
\begin{pmatrix}
\s_1 \\
\vdots\\
\s_{n'}
\end{pmatrix}\in\mbb{R}_{+}^{nn'}
\ee
Using this notation,~\eqref{e0} is equivalent to $A\s\geq \b$. The question of whether there exists a $\s\geq 0$ that satisfies $A\s\geq\b$ is a feasibility problem in linear programming.
From Farkas lemma, such a row vector $\s\in\mbb{R}_{+}^{nn'}$ exists if and only if for any vector $\t\in\mbb{R}_{+}^{(m+1)n'}$ that satisfies $\t^TA\leq 0$ we also have
$\b\cdot\t\leq 0$. Denote $\t=(\t_1,...,\t_{n'},\v)^T$, where for each $x'=1,...,{n'}$, $\t_{x'}\in\mbb{R}_{+}^{m}$ and $\v=(v_1,...,v_{n'})^T\in\mbb{R}^{n'}$. With this notation, the dual problem has the form
\be\label{a0}
v_{x'}\e\geq \t_{x'}^TU^TP\quad\forall x'=1,...,n'\quad\Rightarrow\quad\e'\cdot\v\geq \sum_{x'=1}^{n'}\t_{x'}^TU^T\q_{x'}
\ee
Note that the condition  $v_{x'}\e\geq \t_{x'}^TU^TP$ implies that $v_{x'}\geq \t_{x'}^TU^T\p_{x}$ for all $x=1,...,n$ and $x'=1,...,n'$. 
Therefore, taking the optimal value $v_{x'}=\max_{x}\t_{x'}^TU^T\p_{x}$ we get that~\eqref{a0} holds if and only if
for any $n'$ vectors $\t_1,...,\t_{n'}\in\mbb{R}_{+}^{m}$
\be
\sum_{x'=1}^{n}\max_{x}\t_{x'}^TU^T\p_{x}\geq \sum_{x'=1}^{n}\t_{x'}^TU^T\q_{x'}\;.
\ee
Denoting by $\r_{x'}\equiv U\t_{x'}$ we obtain the the above equation is equivalent to~\eqref{lem0}.
This completes the proof of the lemma.

We are now ready to prove Theorem~\ref{cun}. We therefore assume now that $\mM\precsim \mN$. 
Then, by definition we have for all $\ell$ and any $\r_1,...,\r_\ell\in\mbb{R}_{+}^{m}$ whose components are arranged in a non-decreasing order
\be
\sum_{z=1}^{\ell}\max_{x}\;\r_z\cdot\p_{x}\geq \sum_{z=1}^{\ell}\max_{x'}\;\r_z\cdot\q_{x'}.
\ee
Taking $\ell=n'$ we get that
 \be
\sum_{x'=1}^{n'}\max_{x}\;\r_{x'}\cdot\p_{x}\geq \sum_{z=1}^{n'}\max_{x'}\;\r_z\cdot\q_{x'}\geq \sum_{x'=1}^{n'}\;\r_{x'}\cdot\q_{x'}\;.
\ee   

Therefore, from the two lemmas above it follows that $P$ and $Q$ are related as in~\eqref{rel1}. This completes the ``only if'' direction in the theorem statement. 

Finally, we need to show that if $\mathcal{M} = \mathcal{V}^{ZY' \rightarrow Y} \circ \mathcal{N}^{X' \rightarrow Y'} \circ \mathcal{S}^{X \rightarrow X'Z}$, then  $\mathcal{M} \precsim \mathcal{N}$. We introduce here Dirac notation, such that $x \rightarrow \ket{x}\bra{x}$. Then the probability of observing outcome $y$ when $x$ is sent through $\mathcal{N}$ can be written as $\text{Tr}[\mathcal{N}(\ket{x}\bra{x})\ket{y}\bra{y}]$. Upon substituting everything into the reward function, we have:
\small
\begin{align*}
    & \text{Prob}_{T}(\mathcal{M}) = \text{Prob}_{T}(\mathcal{V}^{ZY' \rightarrow Y} \circ \mathcal{N}^{X' \rightarrow Y'} \circ \mathcal{S}^{X \rightarrow X'Z}) \\
    &= \sum_{z} \max_{x} \Bigg( \sum_{w=1}^{m} t_{w|z}  \sum_{y=1}^{w} \text{Tr}[\mathcal{V} \circ  \mathcal{N}  \circ \mathcal{S}(\ket{x}\bra{x}) \ket{y}\bra{y}] \Bigg) \\
     & \leq \sum_{z} \max_{x} \Bigg( \sum_{w=1}^{m} t_{w|z} \sum_{y=1}^{w} \text{Tr}[ \sum_{z'} p(z'|x) \mathcal{V}_{z'}( \mathcal{N}(\ket{x}\bra{x})) \ket{y}\bra{y}] \Bigg) \\
    & \leq \sum_{z} \max_{x} \Bigg( \sum_{w=1}^{m} t_{w|z} \sum_{y=1}^{w} \text{Tr}[ \max_{\mathcal{V}} \mathcal{V}( \mathcal{N}(\ket{x}\bra{x})) \ket{y}\bra{y}] \Bigg) \\
     & \leq \sum_{z} \max_{x} \Bigg( \sum_{w=1}^{m} t_{w|z} \sum_{y=1}^{w} \text{Tr}[ \mathcal{N}(\ket{x}\bra{x}) \ket{y}\bra{y}] \Bigg) \\
     &= \text{Prob}_{T}(\mathcal{N})
\end{align*}
\normalsize
where the second to last line follows from recalling that $\mathcal{N}$ is already ordered by definition s.t. $\text{Tr}[\mathcal{N}(\ket{x}\bra{x}) \ket{j}\bra{j}] \geq \text{Tr}[\mathcal{N}(\ket{x}\bra{x}) \ket{j+1}\bra{j+1}]$ for all $x$ and $j$. 
\end{proof}
{\bf Operational Interpretation.} Here, we demonstrate that the payoff function $\text{Prob}_{T}(\mathcal{N})$ provides a first operational interpretation of the dynamical monotones introduced in~\cite{Gour3}. \\~\\
We begin by defining a quantum channel $\mathcal{N}$ as a completely-positive trace preserving map from system $A$ to system $B$, which can be written as $\mathcal{N} \in \text{CPTP}(\text{A} \rightarrow \text{B})$. A superchannel $\Theta$ is then defined as a map from quantum channels to quantum channels s.t. $\Theta : \text{CPTP}(\text{A}_{0} \rightarrow \text{A}_{1}) \rightarrow \text{CPTP}(\text{B}_{0} \rightarrow \text{B}_{1})$. 

Finally, recall that the Choi matrix $J_{\mathcal{N}}$ of channel $\mathcal{N}$ is defined as 
\begin{align*}
    J_{\mathcal{N}} \triangleq \sum_{i, j} \ket{i}\bra{j} \otimes \mathcal{N}(\ket{i}\bra{j})
\end{align*}
We now restate the lemma first proved in~\cite{Gour3}.

\begin{lemma}
Let $\text{FREE}(A \rightarrow B)$ be a convex and topologically closed set where $A = (A_{0}, A_{1})$ and $B = (B_{0}, B_{1})$. Denote the Choi matrix for any channel $\mathcal{N}$ as $J_{\mathcal{N}}$
For any quantum channel $\mathcal{P}_{B} \in \text{CPTP}(B_{0} \rightarrow B_{1})$, define
\begin{align*}
f_{\mathcal{P}} (\mathcal{N}) & \triangleq \max_{\Theta \in \text{FREE}(A \rightarrow B)} \big< \mathcal{P}, \Theta [\mathcal{N}_{A}] \big> \\
& = \max_{\Theta \in \text{FREE}(A \rightarrow B)} \text{Tr}\Big( J_{\mathcal{P}}^{\dag} J_{\Theta[\mathcal{N}]} \Big)
\end{align*}
for every $\mathcal{N}_{A} \in \text{CPTP} (A_{0} \rightarrow A_{1})$. Let $\mathcal{N}_{A} \in \text{CPTP}(A_{0} \rightarrow A_{1})$ and $\mathcal{M}_{B} \in \text{CPTP}(B_{0} \rightarrow B_{1})$ be two
quantum channels. Then, $\mathcal{M}_{B} = \Theta_{A\rightarrow B} [\mathcal{N}_{A}]$, for some super-channel $\Theta \in \text{FREE} (A \rightarrow B)$ if and only if
\begin{align*}
f_{\mathcal{P}} (\mathcal{N}_{A}) \geq f_{\mathcal{P}} (\mathcal{M}_{B}) \ \ \ \ \ \ \  \forall P \in \text{CPTP}(B_{0} \rightarrow B_{1})
\end{align*}
\end{lemma}
\begin{proof}
See~\cite{Gour3} for a complete proof. 
\end{proof}

We are now ready to prove how our work provides an operational interpretation of Lemma 9. 
\begin{theorem}
Define the set of free classical superchannels $\text{FREE}(\text{A} \rightarrow \text{B})$ to contain all $\Theta$ of the form:
\begin{align*}
    \Theta[\mathcal{N}](\ket{k}\bra{k}) = \sum_{z} p(z|k) \mathcal{V}_{z}\Big( \mathcal{N} \Big( \mathcal{S}(\ket{k}\bra{k}) \Big) \Big) \ \ \ \forall \ k
\end{align*}
Then, $\mathcal{M}_{B} = \Theta_{A\rightarrow B} [\mathcal{N}_{A}]$, for some super-channel $\Theta \in \text{FREE} (A \rightarrow B)$ if and only if
\begin{align*}
f_{\mathcal{P}} (\mathcal{N}_{A}) \geq f_{\mathcal{P}} (\mathcal{M}_{B}) \ \ \ \ \ \ \  \forall P \in \text{CPTP}(B_{0} \rightarrow B_{1})
\end{align*}
\end{theorem} 
\begin{proof}
From Theorem 2 of this work, $\text{Prob}_{T}(\mathcal{M}) \leq \text{Prob}_{T}(\mathcal{N})$ for all $T$ if and only if there exists a free classical superchannel s.t. $\Theta[\mathcal{N}] = \mathcal{M}$. Thus, the application above will hold if we can demonstrate the following:
\begin{itemize}
    \item  For every $T$ with fixed $z$, there exists a classical channel $\mathcal{P}$ and positive constant $\alpha_{T}$ such that $\text{Prob}_{T}(..) = \alpha_{T} f_{\mathcal{P}}(..)$
    \item For every classical channel $\mathcal{P}$, there exists a distribution $T$ and positive constant $\alpha_{\mathcal{P}}$ such that $f_{\mathcal{P}}(...) = \alpha_{\mathcal{P}} \text{Prob}_{T}(...)$
\end{itemize}
We first prove the second statement. 
\small
\begin{align*}
& f_{\mathcal{P}}(\mathcal{N}) \triangleq \max_{\Theta \in \text{FREE}(\text{A} \rightarrow \text{B})} \text{Tr}\Big( J_{\mathcal{P}}^{\dag} J_{\Theta[\mathcal{N}]} \Big) \\
&=  \max_{\Theta \in \text{FREE}(\text{A} \rightarrow \text{B})} \text{Tr}\Bigg( \Big( \sum_{ij} \mathcal{P}(\ket{i}\bra{j}) \otimes \ket{i}\bra{j} \Big) \Big( \sum_{i'j'} \Theta[\mathcal{N}](\ket{i}\bra{j}) \otimes \ket{i}\bra{j} \Big) \Bigg) \\
&=  \max_{\Theta \in \text{FREE}(\text{A} \rightarrow \text{B})} \text{Tr}\Bigg( \sum_{j} \mathcal{P}(\ket{j}\bra{j}) \otimes \ket{j}\bra{j} \Big( \sum_{j'} \Theta[\mathcal{N}](\ket{j'}\bra{j'}) \otimes \ket{j'}\bra{j'} \Big) \Bigg) \\
&=  \max_{\Theta \in \text{FREE}(\text{A} \rightarrow \text{B})} \text{Tr}\Bigg( \sum_{j} \mathcal{P}(\ket{j}\bra{j}) \Theta[\mathcal{N}](\ket{j}\bra{j}) \Bigg) \\
\end{align*}
\normalsize
where the simplification of the Choi Matrix follows from noting that both $\mathcal{P}$ and $\mathcal{N}$ are classical channels. Recalling the definition of free operations, the above becomes 
\small
\begin{align*}
   f_{\mathcal{P}}(\mathcal{N}) &= \max_{\mathcal{S}, \{\mathcal{V}_{z}\}, \{p(k|z)\}} \sum_{k} \text{Tr}\Bigg[ \mathcal{P}(\ket{k}\bra{k}) \sum_{z} p(z|k) \mathcal{V}_{z}\Bigg(\mathcal{N} \Big(  \mathcal{S} (\ket{k}\bra{k})\Big) \Bigg) \Bigg] \\
    &=  \max_{\mathcal{S}, \{\mathcal{V}_{z}\}, \{p(k|z)\}} \sum_{k} \sum_{z} p(z|k) \\
    & \times \text{Tr}\Bigg[ \mathcal{P}(\ket{k}\bra{k}) \mathcal{V}_{z}\Bigg(\mathcal{N} \Big(  \mathcal{S} (\ket{k}\bra{k})\Big) \Bigg) \Bigg] \\
    &=  \max_{\mathcal{S}, \{V_{k} \}} \sum_{k} \text{Tr}\Bigg[ \mathcal{P}(\ket{k}\bra{k}) \mathcal{V}_{k}\Bigg(\mathcal{N} \Big(  \mathcal{S} (\ket{k}\bra{k})\Big) \Bigg) \Bigg] \\
\end{align*}
\normalsize
The third line follows by setting $p(z|k) = \delta_{z, k}$ and noting that this will always be optimal, as 
\small
\begin{align*}
\sum_{z} p(z|k)\text{Tr}\Bigg[ \mathcal{P}(\ket{k}\bra{k}) \mathcal{V}_{z}\Bigg(\mathcal{N} \Big(  \mathcal{S} (\ket{k}\bra{k})\Big) \Bigg) \Bigg] \\
\leq \max_{z} \Big( \text{Tr}\Bigg[ \mathcal{P}(\ket{k}\bra{k}) \mathcal{V}_{z}\Bigg(\mathcal{N} \Big(  \mathcal{S} (\ket{k}\bra{k})\Big) \Bigg) \Bigg] \Big)
\end{align*} 
\normalsize
We continue simplifying this expression: 
\begin{align*}
    f_{\mathcal{P}}(\mathcal{N}) &= \max_{\mathcal{S}, \{V_{k} \}} \sum_{k=1}^{d} \text{Tr}\Bigg[ \mathcal{P}(\ket{k}\bra{k}) \mathcal{V}_{k}\Bigg(\mathcal{N} \Big(  \mathcal{S} (\ket{k}\bra{k})\Big) \Bigg) \Bigg] \\
   &= \max_{\mathcal{S}, \{V_{k} \}} \sum_{k=1}^{d} \sum_{y=1}^{d} q_{y|k} \bra{y} \mathcal{V}_{k}\Bigg(\mathcal{N} \Big(  \mathcal{S} (\ket{k}\bra{k})\Big) \Bigg) \ket{y}\\
    &= \max_{\mathcal{S}} \sum_{k=1}^{d} \sum_{y=1}^{d} q^{\downarrow}_{y|k} \bra{y} \mathcal{N}^{\downarrow} \Big(  \mathcal{S} (\ket{k}\bra{k})\Big) \ket{y}\\
    &=  \sum_{k=1}^{d}  \max_{x} \sum_{y=1}^{d} q^{\downarrow}_{y|k} \bra{y} \mathcal{N}^{\downarrow} (\ket{x}\bra{x}) \ket{y}\\
    &=  \sum_{k=1}^{d}  \max_{x} \sum_{y=1}^{d} q^{\downarrow}_{y|k} p^{\downarrow}_{y|x} 
\end{align*}
    
where $q_{y|k} \triangleq \bra{y}\mathcal{P}(\ket{k}\bra{k})\ket{y}$ and $p^{\downarrow}_{y| x} \triangleq \bra{y} \mathcal{N}^{\downarrow} (\ket{x}\bra{x}) \ket{y}$. Finally, we manipulate the above expression to rewrite it in the form of $\text{Prob}_{T}(\mathcal{N})$:
    
\begin{align*}
    f_{\mathcal{P}}(\mathcal{N})&= \sum_{k=1}^{d} \max_{x} \Bigg( q^{\downarrow}_{1|k} p^{\downarrow}_{1|x} + \sum_{y=2}^{d} q^{\downarrow}_{y|k} (\sum_{j=1}^{y} p^{\downarrow}_{j|x} - \sum_{j=1}^{y-1} p^{\downarrow}_{j|x} ) \Bigg) \\
    &= \sum_{k=1}^{d} \max_{x} \Bigg( q^{\downarrow}_{1|k} p^{\downarrow}_{1|x} + \sum_{y=2}^{d} q^{\downarrow}_{y|k} \sum_{j=1}^{y} p^{\downarrow}_{j|x} - \sum_{y=1}^{d} q^{\downarrow}_{y+1|k} \sum_{j=1}^{y} p^{\downarrow}_{j|x} ) \Bigg) \\
    &= \sum_{k=1}^{d} \max_{x} \sum_{y=1}^{d} (q^{\downarrow}_{y|k} - q^{\downarrow}_{y+1|k}) \sum_{j=1}^{y} p^{\downarrow}_{j|x} \\
    &= \sum_{k=1}^{d} |\tilde{\mathbf{t}}_{k}| \max_{x} \sum_{y=1}^{d} \frac{{\tilde{t}}_{y|k}}{|\tilde{\mathbf{t}_{k}}|} \sum_{j=1}^{y} p^{\downarrow}_{j|x} \\
    &= \sum_{k} |\mathbf{\tilde{t}}_{k}| \text{Prob}_{\frac{\tilde{\mathbf{t}}(k)}{|\tilde{\mathbf{t}}(k)|}}(\mathcal{N})
\end{align*}
 where $\tilde{\mathbf{t}}_{k} = \{ (q^{\downarrow}_{1|k} - q^{\downarrow}_{2|k}), (q^{\downarrow}_{2|k}-q^{\downarrow}_{3|k}), ..., (q^{\downarrow}_{d-1|k}-q^{\downarrow}_{d|k}), (q^{\downarrow}_{d|k} - 0) \}$ and in line two we use the definition $q^{\downarrow}_{d+1|k} = 0$ (as $q^{\downarrow}_{y|k}$ can always be padded with extra zeros).
 
 
 Finally, we relabel the dummy variable $k$ to $z$ and simplify:
\begin{align*}
    f_{\mathcal{P}}(\mathcal{N}) &= \sum_{z} |\mathbf{\tilde{t}}_{z}| \text{Prob}_{\frac{\tilde{\mathbf{t}}_{z}}{|\tilde{\mathbf{t}}_{z}|}}(\mathcal{N}) \\
    &= (\sum_{z'} |\mathbf{\tilde{t}}_{z'}|) \sum_{z} \frac{|\tilde{\mathbf{t}}_{z}|}{\big( \sum_{z'}  |\mathbf{\tilde{t}}_{z'}| \big)} \text{Prob}_{\frac{\tilde{\mathbf{t}}_{z}}{|\tilde{\mathbf{t}}_{z}|}}(\mathcal{N}) \\
    &= (\sum_{z'} |\mathbf{\tilde{t}}_{z'}|) \text{Prob}_{T}(\mathcal{N})
\end{align*}
where $T$ is defined by its elements:
\begin{align*}
    t_{k, z} &= \frac{|\tilde{\mathbf{t}}_{z}|}{\sum_{z'}  |\mathbf{\tilde{t}}_{z'}|} \times \frac{\tilde{t}^{\downarrow}_{k|z}}{|\tilde{\mathbf{t}}_{z}|} = \frac{\tilde{t}^{\downarrow}_{k|z}}{\sum_{z'}  |\mathbf{\tilde{t}}_{z'}|}
\end{align*}

Finally, we need to show that for every $T$ with fixed $z=1$ (i.e. $T$ is a vector), there exists a classical channel $\mathcal{P}$ and positive constant $\alpha_{T}$ such that $\text{Prob}_{T}(\mathcal{N}) = \alpha_{T} f_{\mathcal{P}}(\mathcal{N})$. However, from the above expression it immediately follows that we can select $\mathbf{\tilde{t}}_{z=1}$ s.t. $ T = \frac{\tilde{\mathbf{t}}_{z=1}}{|\tilde{\mathbf{t}}_{z=1}|}$.  Then the channel $\mathcal{P}$ corresponding to $\tilde{\mathbf{t}}$ will satisfy $\text{Prob}_{T}(\mathcal{N}) = \alpha_{T} f_{\mathcal{P}}(\mathcal{N})$.

\end{proof}

\section{Proof of Lemma 3}
{\bf Lemma 3.} \emph{ The pre-order induced by the set of classical gambling games is unchanged if we restrict to the set of gambling games such that $\mathcal{T} = \mathbf{p}$ where $\mathbf{p}$ is a vector probability distribution. Equivalently, $\mathcal{M} \precsim \mathcal{N}$ if and only if $\text{Prob}_{\mathbf{p}}(\mathcal{M}) \leq \text{Prob}_{\mathbf{p}}(\mathcal{N})$ for all $\mathbf{p}$.}
\begin{proof}
Referring to the definition of $\text{Prob}_{\mathcal{T}}(\mathcal{N})$ from the previous section, 
\begin{align*}
    \text{Prob}_{\mathcal{T}}(\mathcal{N}) &= \sum_{z=1}^{\ell} \max_{x} \Big( \sum_{y=1}^{m} \sum_{w=y}^{m} t_{wz} p_{y|x} \Big) \\
    &= \sum_{z=1}^{\ell} t_{z} \max_{x} \Big( \sum_{y=1}^{m} \sum_{w=y}^{m} t_{w|z} p_{y|x} \Big) \\
    &= \sum_{z=1}^{\ell} t_{z} \text{Prob}_{\{\mathbf{t}_{z}\}}(\mathcal{N})
\end{align*}
where $t_{kz} = t_{z} t_{w|z}$ such that $\{\mathbf{t}_{z}\} = \{t_{1|z}, ..., t_{m|z}\}$. The statement then follows. 
\end{proof}

\section{Details of Remark 2}
We here introduce a lemma, from which remark 2 follows.

\begin{lemma} The set of random unitary superchannels are equivalent to the set of uniformity-preserving superchannels. Equivalently, $\Theta$ is a classical uniformity-preserving superchannel if and only if $\Theta$ can be written as:
\begin{align*}
    \Theta[\mathcal{N}](\ket{x}\bra{x}) = \sum_{z} p(z|x) \mathcal{V}_{z}\Big( \mathcal{N} \Big( \mathcal{S}(\ket{x}\bra{x}) \Big)\Big) \ \ \ \forall \ x
\end{align*}
\end{lemma}
\begin{proof}
First, we show that the above form is uniformity preserving

\begin{align*}
    \Theta[\mathcal{R}](\ket{x}\bra{x}) &= \sum_{z} p(z|x) \mathcal{V}_{z}\Big(\mathcal{R}(\mathcal{S}(\ket{x}\bra{x})) \Big) \ \ \ \forall x \\
    &=  \sum_{z} p(z|x) \mathcal{V}_{z}\Big( \frac{\mathbb{I}}{|B|}  \Big) \ \ \ \forall x \\
        &=  \sum_{z} p(z|x) \frac{\mathbb{I}}{|B|} \ \ \ \forall x \\
        &= \frac{\mathbb{I}}{|B|}  \ \ \ \forall x \\
\end{align*}

Next, we show that any classical uniformity preserving superchannel $\Theta$ can be written in the above form.  First, we note that the most general classical superchannel can be written as 

\begin{align*}
    \Theta[\mathcal{N}](\ket{x}\bra{x}) = \sum_{z} p(z|x) \mathcal{P}\Big( \mathcal{N} \Big( \mathcal{S}(\ket{x}\bra{x}) \Big)  \otimes \ket{z}\bra{z}  \Big) \ \ \ \forall \ x
\end{align*}

where $\mathcal{P}_{z}$ is an arbitrary post-processing channel. (Depicted in the following figure with dummy variable $z \rightarrow w'$):
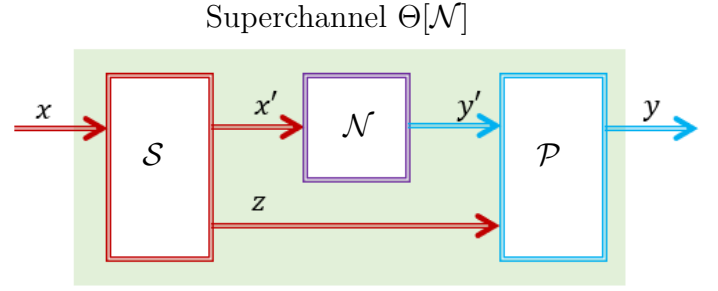
\begin{figure}[h]
\large
\begin{overpic}[scale=.87]{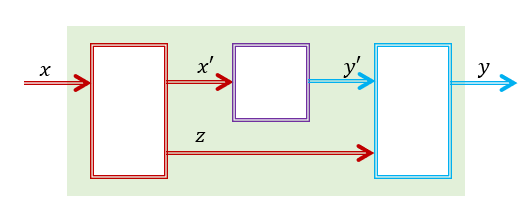}
\put(61,53){$\mathcal{S}$}
\put(137,64){$\mathcal{N}$}
\put(210,53){$\mathcal{P}$}
\put(85, 105){Superchannel $\Theta[\mathcal{N}]$}
\end{overpic}
\caption{\linespread{1}\selectfont{\small General classical superchannel. }} 
  \label{channel3}
\end{figure}

Upon imposing the restraint that $\Theta$ be uniformity-preserving, we have

\begin{align*}
     \frac{\mathbb{I}}{|B|} &= \Theta[\mathcal{R}] \big( \ket{x}\bra{x} \big) \ \ \ \forall x \\
    &= \sum_{z} p(z|x) \mathcal{P}\Big( \mathcal{R} \Big( \mathcal{S}(\ket{x}\bra{x}) \Big)  \otimes \ket{z}\bra{z}  \Big) \ \ \ \forall \ x \\
    &= \sum_{z} p(z|x) \mathcal{P}\Big( \frac{\mathbb{I}}{|B|}  \otimes \ket{z}\bra{z}  \Big) \ \ \ \forall x \\
    &= \sum_{z} p(z|x) \mathcal{P}_{z}\Big( \frac{\mathbb{I}}{|B|}  \Big) \ \ \ \forall x 
\end{align*}
where $\mathcal{P}_{z}(\rho) \triangleq \mathcal{P}(\rho \otimes \ket{z}\bra{z})$. Finally, denote by $P^{x}$ the transition probability matrix corresponding to the channel $\sum_{z} p(z|x) \mathcal{P}_{z}$. Clearly, $\sum_{z} p(z|x) \mathcal{P}_{z}$ must be unital for all $x$ s.t. 
\begin{align*}
    P^{x} \mathbf{u}_{B} = \mathbf{u}_{B}
\end{align*}
where $\mathbf{u}_{B}$ is the uniform distribution for system $B$. It follows from the above that $P^{x}$ is doubly stochastic for all $x$, and therefore (from Birkhoff's lemma), can be written as a convex combination of permutation matrices. 

Then there exists a set of permutation matrices $\{V_{w, x}\}$ corresponding to isometry channels $\{\mathcal{V}_{w, x}\}$ and probability distribution $\sum_{w} q_{w|x} = 1$ for all $x$ s.t. 
\begin{align*}
    \sum_{z} p(z|x) \mathcal{P}_{z} & \triangleq \sum_{w} q_{w| x} \mathcal{V}_{w, x} \ \ \ \forall x
\end{align*}
Finally, we can then rewrite our superchannel $\Theta$ as:

\begin{align*}
    \Theta[\mathcal{N}](\ket{x}\bra{x}) &= \sum_{z} p(z|x) \mathcal{P}_{z}\Big( \mathcal{N} \Big( \mathcal{S}(\ket{x}\bra{x}) \Big) \Big) \ \ \ \forall \ x \\
    &= \sum_{t} q_{t} \mathcal{V}_{t}\Big( \mathcal{N} \Big( \mathcal{\tilde{S}}(\ket{x}\bra{x}) \Big) \Big)
\end{align*}
where $\tilde{S}$ is related to the original pre-processing channel as shown in Fig. 6 and $t = (w, z)$.

\begin{figure}[h]
\large
\begin{overpic}[scale=.84]{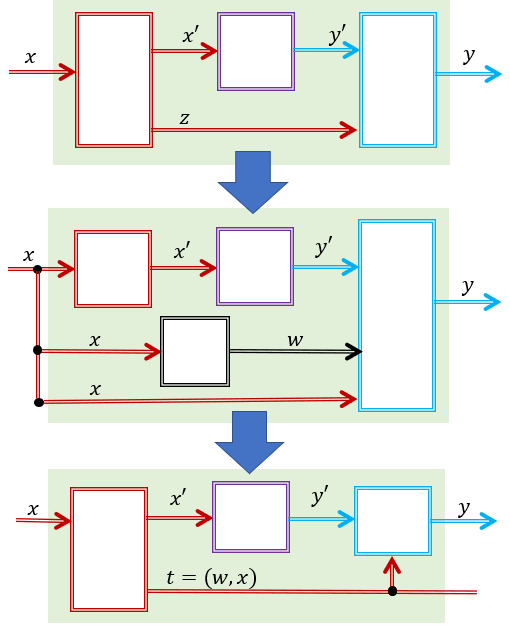}
\put(54,273){$\mathcal{S}$}
\put(125,289){$\mathcal{N}$}
\put(197,275){$\mathcal{P}$}
\put(53,176){$\mathcal{S}_{1}$}
\put(125,178){$\mathcal{N}$}
\put(190,160){$\mathcal{V}_{w, x}$}
\put(95,137){$\mathbf{q}$}
\put(52,37){$\tilde{\mathcal{S}}$}
\put(122,53){$\mathcal{N}$}
\put(194,52){$\mathcal{V}_{t}$}
\end{overpic}
  \label{channel3}
  \caption{Restructuring a general uniformity preserving superchannel.}
\end{figure}
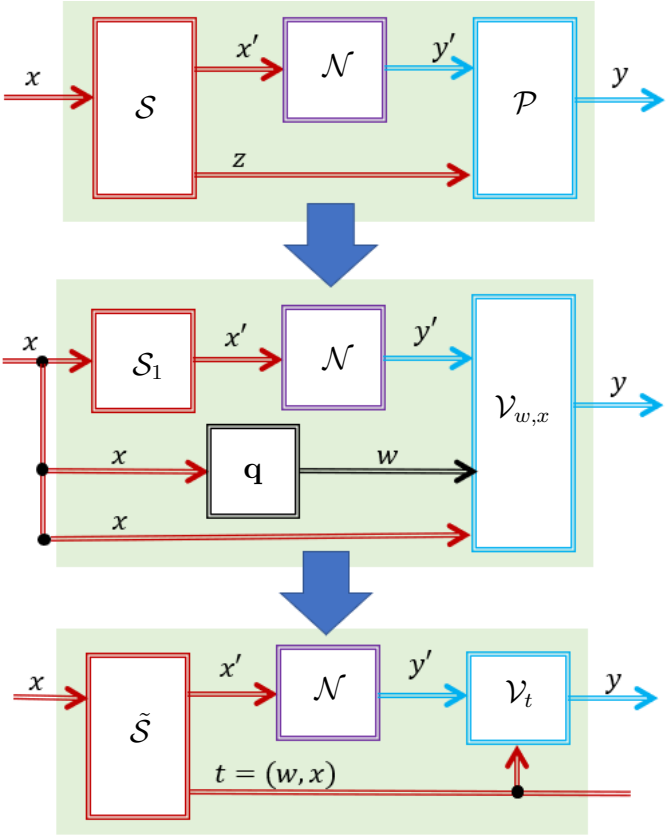

\end{proof}

\section{Proof of Theorem 3}
{\bf Theorem 3.}\emph{
Let $H$ be a classical dynamical entropy that reduces to the Shannon entropy, $H_{S}$, on classical states. Then for all $\mathcal{N} \in \mathfrak{L}({\text{A} \rightarrow \text{B}})$, 
\begin{align*}
    H(\mathcal{N}) &= \min_{x} H_{S}\Big(\mathcal{N}(\ket{x}\bra{x})\Big)
\end{align*}
}
\begin{proof}

 We first refer to previous literature on extensions of channel divergences~\cite{Tomamichal, Gour4}. From this, we define the minimal extension $\underline{D}$ and maximal extension $\overline{D}$ of any channel divergence $\mathbb{D}$ which reduces to the Kullback-Leibler divergence $D$ on states as:
 \begin{align*}
     \underline{D}(\mathcal{N} \big| \big| \mathcal{M}) & \triangleq \max_{x} \Big( D \Big(\mathcal{N}\big(\ket{x}\bra{x}) \Big| \Big| \mathcal{R} (\ket{x}\bra{x}) \Big) \Big) \\
     \overline{D}(\mathcal{N} \big| \big| \mathcal{M}) &=  \inf_{\mathbf{p}, \mathbf{q}} \Big( D(\mathbf{p} \big| \big| \mathbf{q}) \Big)
 \end{align*}
 where the infimum in the second line is over all $(\mathbf{p}, \mathbf{q})$ s.t. $(\mathbf{p}, \mathbf{q}) \prec
\big( \mathcal{N}(\ket{x}\bra{x}), \mathcal{M}(\ket{x}\bra{x}) \big)$ for all $x$. We may then define the corresponding regularized extensions as 
\begin{align*}
    \underline{D}^{\text{reg}}(\mathcal{N} \big| \big| \mathcal{M}) &= \lim_{k \rightarrow \infty} \frac{1}{k} \underline{D}(\mathcal{N}^{\otimes k } \big| \big| \mathcal{M}^{\otimes k}) \\ 
    \overline{D}^{\text{reg}}(\mathcal{N} \big| \big| \mathcal{M}) &= \lim_{k \rightarrow \infty} \frac{1}{k} \overline{D}(\mathcal{N}^{\otimes k } \big| \big| \mathcal{M}^{\otimes k})
\end{align*}

This allows us to then define a minimal extension $\underline{H}$ and maximal extension $\overline{H}$ of $H$:
\begin{align*}
   \underline{H}(\mathcal{N}) & \triangleq \text{log}|B| - \overline{D}^{\text{reg}}(\mathcal{N} \big| \big| \mathcal{R}) \\
     \overline{H}(\mathcal{N}) & \triangleq \text{log}|B| - \underline{D}(\mathcal{N} \big| \big| \mathcal{R}) \\
\end{align*}
Since  $H$ reduces to the Shannon entropy on classical states, then $\text{log}|B|-H(\mathcal{N})$ reduces to $D(\mathcal{N} \big| \big| \mathcal{R})$.
In~\cite{Gour4}, it was shown that $ \underline{D}(\mathcal{N}) \leq \mathbb{D}(\mathcal{N}) \leq \overline{D}^{\text{reg}}(\mathcal{N})$, from which it then follows that 
\begin{align*}
    \underline{H}(\mathcal{N}) \leq H(\mathcal{N}) \leq \overline{H}(\mathcal{N})
\end{align*}
It was additionally shown in~\cite{Gour4} that
\begin{align*}
\underline{D}(\mathcal{N} \big| \big| \mathcal{R}) &= \overline{D}^{\text{reg}}(\mathcal{N} \big| \big| \mathcal{R})
\end{align*}
Thus, 
\begin{align*}
    \underline{H}(\mathcal{N}) \leq H(\mathcal{N}) \leq \overline{H}(\mathcal{N}) = \underline{H}(\mathcal{N})
\end{align*}
and so
\begin{align*}
    H(\mathcal{N}) &= \text{log}|B| - \max_{x} \Big( D \Big(\mathcal{N}\big(\ket{x}\bra{x}) \big| \big| \mathcal{R} (\ket{x}\bra{x}) \Big) \Big) \\
    &= \min_{x} H_{S}(\mathcal{N}(\ket{x}\bra{x}))
\end{align*}
where $H_{S}$ is the Shannon entropy.

Finally, we note that $H(\mathcal{N}) = \min_{x} H_{S}\big( \mathcal{N}(\ket{x}\bra{x}) \big)$ is an asymptotically continuous entropy function. It has previously been shown~\cite{Andreas_Winter} that the following function is asymptotically continuous:

\begin{align*}
f(\mathcal{N}^{\text{A} \rightarrow \text{B}}) & \triangleq \max_{\psi_{RA} \in \mathcal{D}(RA)} D(\mathcal{N}(\psi_{RA}) \big| \big| \mathcal{R}(\psi_{RA}))
\end{align*}

where $D(\rho \big| \big| \sigma) = \text{Tr}[\rho (\text{log} \rho - \text{log} \sigma)]$ and where $\mathcal{D}(RA)$ denotes the set of density matrices corresponding to systems $RA$. \\~\\
In the case where $\mathcal{N}$ is a classical channel, the maximum can be achieved with a trivial reference system, such that $f$ becomes:

\begin{align*}
f(\mathcal{N}) & = \max_{ \psi_{A} \in \mathcal{D}(A)} D \Big(\mathcal{N}(\psi_{A}) \big| \big| \mathcal{R}(\psi_{A}) \Big) \\
& =  \max_{ x } D \Big(\mathcal{N}(\ket{x}\bra{x}) \big| \big| \mathcal{R}(\ket{x}\bra{x}) \Big) \\
\end{align*}
where in the second line $D$ reduces to the Kullback-Leibler divergence, and we note that it is sufficient to consider pure classical states when $\mathcal{N}$ is a classical channel.  Thus, $H(\mathcal{N}) = \text{log}|B| - f(\mathcal{N})$ and $H$ is the unique asymptotically continuous channel entropy function.


\end{proof}


\begin{thebibliography}{10}

\bibitem{entropy_econ}
A.~Jakimowicz, ``{The Role of Entropy in the Development of Economics},'' {\em
  Entropy, 22, 452}, 2020.

\bibitem{Bekenstein1}
J.~D. Bekenstein, ``{Black holes and the second law},'' {\em Lettere al Nuovo
  Cimento, 4, 737}, 1972.

\bibitem{Bekenstein2}
J.~D. Bekenstein, ``{Black holes and entropy},'' {\em Physical Review D, 7,
  2333}, 1973.

\bibitem{arrow_of_time}
M.~Mackey, ``{Time's Arrow: The Origins of Thermodynamic Behavior},'' {\em
  Berlin Heidelberg New York: Springer}, 1992.

\bibitem{Jaynes}
E.~Jaynes, ``{Gibbs vs Boltzmann Entropies},'' {\em American Journal of
  Physics, Vol. 33, No. 5, 391-398}, 1965.

\bibitem{Tsallis}
E.~Jaynes, ``{Possible generalization of Boltzmann-Gibbs statistics},'' {\em J
  Stat Phys 52, 479–487}, 1988.

\bibitem{Renyi}
A.~Renyi, ``{On measures of information and entropy},'' {\em Proceedings of the
  fourth Berkeley Symposium on Mathematics, Statistics and Probability}, 1960.

\bibitem{VonNeumann}
I.~Bengtsson and K.~Zyczkowski, ``{Geometry of Quantum States: An Introduction
  to Quantum Entanglement},''

\bibitem{Shannon}
C.~E. Shannon, ``{A Mathematical Theory of Communication},'' {\em The Bell
  System Technical Journal, Vol. 27, pp. 379–423, 623–656}, 1948.

\bibitem{MolarEntropy}
K.~Kosanke, ``{Chemical Thermodynamics},'' {\em Journal of Pyrotechnics. p.
  29}, 2004.

\bibitem{MixingEntropy}
I.~Prigogine, ``{Introduction to Thermodynamics of Irreversible Processes,
  third edition},'' {\em Interscience Publishers, New York, p. 12.}, 1967.

\bibitem{LoopEntropy}
L.~Wang, E.~V. Rivera, M.~G. Benavides-Garcia, and B.~T. Nall, ``{Loop Entropy
  and Cytochrome C Stability},'' {\em J Mol Biol.;353(3):719-29.}, 2005.

\bibitem{guessinggames}
A.~Mendes and K.~E. Morrison, ``{Guessing games},'' {\em Amer. Math. Monthly
  121 33-44}, 2014.

\bibitem{Gour3}
G.~Gour and M.~M. Wilde, ``{Entropy of a quantum channel},'' {\em
  arXiv:1808.06980v2}, 2018.

\bibitem{Yuan}
X.~Yuan, ``{Relative entropies of quantum channels with applications in
  resource theory},'' {\em Phys. Rev. A 99, 032317}, 2019.

\bibitem{Fang}
K.~Fang, O.~Fawzi, R.~Renner, and D.~Sutter, ``{A chain rule for the quantum
  relative entropy},'' {\em Phys. Rev. Lett. 124, 100501}, 2020.

\bibitem{Devetak}
I.~Devetak, C.~King, M.~Junge, and M.~B. Ruskai, ``{Multiplicativity of
  completely bounded p-norms implies a new additivity result},'' {\em
  Communications in Mathematical Physics}, 2006.

\bibitem{DDI}
M.~Dall'Arno, A.~Ho, F.~Buscemi, and V.~Scarani, ``{Data-driven inference and
  observational completeness of quantum devices},'' {\em Phys. Rev. A 102,
  062407}, 2020.

\bibitem{Brandsen}
M.~Dall'Arno, S.~Brandsen, and F.~Buscemi, ``{Device-independent tests of
  quantum channels},'' {\em Proc. R. Soc. A, 473, 20160721}, 2017.

\bibitem{Gour2}
G.~Gour and C.~M. Scandolo, ``{Dynamical Resources},'' {\em
  arXiv:2101.01552v1}, 2020.

\bibitem{Gour}
G.~Gour, A.~Grudka, M.~Horodecki, W.~Klobus, J.~Lodyga, and V.~Narasimhachar,
  ``{The Conditional Uncertainty Principle},'' {\em Phys. Rev. A 97, 042130},
  2018.

\bibitem{Gour5}
G.~{Gour}, ``Comparison of quantum channels by superchannels,'' {\em IEEE
  Transactions on Information Theory}, vol.~65, no.~9, pp.~5880--5904, 2019.

\bibitem{Andreas_Winter}
G.~Gour and A.~Winter, ``How to quantify a dynamical quantum resource,'' {\em
  Physical Review Letters}, vol.~123, Oct 2019.

\bibitem{Tomamichal}
G.~Gour and M.~Tomamichel, ``Optimal extensions of resource measures and their
  applications,'' {\em Physical Review A}, vol.~102, Dec 2020.

\bibitem{Gour4}
G.~Gour, ``{Uniqueness and Optimality of Dynamical Extensions of
  Divergences},'' {\em PRX Quantum 2, 010313}, 2021.

\end{thebibliography}
\end{document}